\documentclass[12pt]{article}
\usepackage{amssymb,amsthm,amsmath,latexsym}
\usepackage{url}
\usepackage{color}
\usepackage{comment}
\usepackage{lscape}
\newtheorem{thm}{Theorem}
\newtheorem{prop}[thm]{Proposition}
\newtheorem{lem}[thm]{Lemma}
\newtheorem{cor}[thm]{Corollary}
\theoremstyle{remark}

\newcommand{\FF}{\mathbb{F}}
\newcommand{\ZZ}{\mathbb{Z}}
\newcommand{\RR}{\mathbb{R}}

\DeclareMathOperator{\wt}{wt}

\DeclareMathOperator{\inv}{inv}

\begin{document}
\title{
Self-dual codes over $\FF_5$ and $s$-extremal unimodular lattices
}

\author{
Masaaki Harada\thanks{
Research Center for Pure and Applied Mathematics,
Graduate School of Information Sciences,
Tohoku University, Sendai 980--8579, Japan.
email: \texttt{mharada@tohoku.ac.jp}.}
}

\maketitle

\begin{abstract}
New $s$-extremal extremal unimodular lattices in dimensions $38$, $40$, $42$
and $44$ are constructed from self-dual codes over $\FF_5$ 
by Construction~A\@.
% In the process of constructing these codes, 
% we obtain a self-dual $[44,22,14]$ code over $\FF_5$, which
% has larger minimum weight than the previously known $[44,22]$ codes.
% In addition, the new self-dual code implies
% a $[43,22,13]$ code over $\FF_5$, which has larger minimum weight than the 
% previously known $[43,22]$ codes.
In the process of constructing these codes, 
we obtain a self-dual $[44,22,14]$ code over $\FF_5$.
In addition, the code implies
a $[43,22,13]$ code over $\FF_5$.
These codes have larger minimum weights than the 
previously known $[44,22]$ codes and $[43,22]$ codes, respectively.
\end{abstract}

%%%%%%%%%%%%%%%%%%%%%%%%%%%%%%
\section{Introduction}\label{Sec:Intro}

%%% code %%%
A (linear) $[n,k]$ code $C$ over $\FF_5$ is a $k$-dimensional subspace of 
$\FF_5^n$, where $\FF_5$ is the finite field of order $5$.
The {dual} code $C^{\perp}$ of $C$ is defined as
$\{x \in \FF_5^n \mid x \cdot y = 0 \text{ for all } y \in C\}$,
where $x \cdot y$ is the standard inner product.
A code $C$ is  called {self-dual} if $C = C^{\perp}$. 
%It is trivial that there is a self-dual $[n,k]$ code if and only if $k=n/2$.
%In addition, 
%for $p \equiv 1 \pmod 4$ a self-dual $[n,n/2]$ code over $\FF_p$ exists
%if and only if $n$ is even, and for $p \equiv 3 \pmod 4$, 
%a self-dual  $[n,n/2]$
%code over $\FF_p$ exists if and only if $n \equiv 0 \pmod4$.
Self-dual codes are one of the most interesting classes of codes.
This interest is justified by many combinatorial objects
and algebraic objects related to self-dual codes.
Unimodular lattices are one of the objects related to self-dual codes.
In addition, 
there are many similarities between self-dual codes and 
unimodular lattices (see~\cite{SPLAG}).
Construction A, which is the most important construction of unimodular lattices
from self-dual codes, gives some similarities.

% This paper has a two-fold purpose
% The purpose of this paper is two-fold.
Extremal unimodular lattices are
unimodular lattices meeting a certain upper bound on minimum norms,
and 
$s$-extremal unimodular lattices are
odd unimodular lattices meeting a certain upper bound on minimum norms
of the shadows. % of odd unimodular lattices.
In this paper,
new $s$-extremal extremal unimodular lattices in dimensions $38,40,42$
and $44$
are constructed from self-dual codes over $\FF_5$ 
by Construction~A\@.
In the process of constructing the above self-dual codes, 
we obtain a self-dual $[44,22,14]$ code over $\FF_5$.
This code has larger minimum weight than the 
previously known self-dual $[44,22]$ codes.
It is a fundamental problem to determine the largest minimum weight 
$d_5(n)$ among all self-dual  $[n,n/2]$ codes over $\FF_5$ for
a given $n$.
Much work has been done concerning this fundamental problem
(see e.g.~\cite{GaOt}, \cite{GL}, \cite{GG}, \cite{GH-GF5}, \cite{GH08},
\cite{GHM} and \cite{LPS-GF5}).
In addition, the new self-dual code implies
a $[43,22,13]$ code over $\FF_5$, which has larger minimum weight than the 
previously known $[43,22]$ codes.

This paper is organized as follows.
In Section~\ref{sec:2},
we give some definitions, notations and basic results used in this paper.
Two methods for constructing self-dual codes are given, namely, 
quasi-twisted codes and four-negacirculant codes.
In Section~\ref{sec:44},
new $s$-extremal extremal unimodular lattices in dimension $44$ are constructed from 
four-negacirculant self-dual codes over $\FF_5$ by Construction~A\@.
In the process of constructing these codes, 
we obtain a self-dual $[44,22,14]$ code $C_{44}$ over $\FF_5$.
This code has larger minimum weight than the 
previously known self-dual codes.
We emphasize that $C_{44}$ is the first example as not only
self-dual $[44,22,14]$ codes but also (linear) $[44,22,14]$ codes.
In addition, $C_{44}$ implies
a $[43,22,13]$ code over $\FF_5$, which has larger minimum weight than the 
previously known $[43,22]$ codes.
Section~\ref{sec:44} also lists the current information on
the largest minimum weight among self-dual $[n,n/2]$ codes over $\FF_5$ for
$22 \le n \le 72$.
In Section~\ref{sec:lat},
new $s$-extremal extremal unimodular lattices in dimensions $38,40$ and $42$
are constructed from self-dual codes over $\FF_5$
by Construction~A\@.
These self-dual codes are constructed as
quasi-twisted codes or four-negacirculant codes.
As a summary, 
we list the current information on the existence of 
non-isomorphic
$s$-extremal unimodular lattices with minimum norm $4$.

All computer calculations in this paper were
done by programs in \textsc{Magma}~\cite{Magma}.

%%%%%%%%%%%%%%%%%%%%%%%%%%%%%%%%%%%%%%%%%%%%%%%%
\section{Preliminaries}\label{sec:2}

In this section, we give definitions, notations and
basic results used in this paper.

%%%%%%%%%%%%%%%%%%%%%%%%%%%%%
\subsection{Unimodular lattices}
A (Euclidean) lattice $L \subset \RR^n$ 
in dimension $n$
is \emph{unimodular} if
$L = L^{*}$, where
the dual lattice $L^{*}$ of $L$ is defined as
$\{ x \in {\RR}^n \mid x \cdot y \in \ZZ \text{ for all }
y \in L\}$ under the standard inner product $ x\cdot  y$.
%% A (Euclidean) integral lattice $L \subset \RR^n$ 
%% in dimension $n$
%% is \emph{unimodular} if $L = L^{*}$, where
%% $L^{*}$ is the dual lattice under the standard inner 
%% product $(x,y)$.
A unimodular lattice $L$ is \emph{even} 
if the norm $x \cdot x$ of every vector $x$ of $L$ is even,
and \emph{odd} otherwise.
An even unimodular lattice in dimension $n$
exists if and only if $n$ is divisible by eight, while
an odd  unimodular lattice exists for every dimension.
The \emph{minimum norm} $\min(L)$ of 
a unimodular lattice $L$ is the smallest 
norm among all nonzero vectors of $L$.
The \emph{kissing number} of $L$ is the number of vectors of minimum norm in $L$.
%%= \sum_{m=0}^{\infty} N_m q^m,
%%where $N_m$ is the number of vectors of norm $m$.
Two lattices $L$ and $L'$ are \emph{isomorphic}, denoted $L \cong L'$,
if there is an orthogonal matrix $A$ with
$L' = \{xA \mid x \in L\}$.
%The \emph{automorphism group} $\Aut(L)$ of $L$ is the group of all
%orthogonal matrices $A$ with $L = L \cdot A$.

Let $L$ be an odd unimodular lattice in dimension $n$.
The \emph{shadow} $S(L)$ of $L$ is defined to be $S(L)= L_0^* \setminus L$,
where $L_0$ denotes the even sublattice of $L$.
Shadows of odd unimodular lattices appeared in~\cite{CS98} 
(see also~\cite[p.~440]{SPLAG}),
and shadows play an important role in the study of odd unimodular lattices.
The {\em theta series} of an odd unimodular lattice $L$ and its shadow $S(L)$
are the formal power series
$\theta_{L}(q) = \sum_{x \in L} q^{x \cdot x}$ and
$\theta_{S(L)}(q) = \sum_{x \in S(L)} q^{x \cdot x}$, respectively.
Conway and Sloane~\cite{CS98} showed that
when the theta series of an odd unimodular lattice $L$
in dimension $n$
is written as:
\begin{equation}
\label{eq:theta}
 \sum_{j =0}^{\lfloor n/8\rfloor} a_j\theta_3(q)^{n-8j}\Delta_8(q)^j,
\end{equation}
the theta series of the shadow $S(L)$
is written as:
\begin{equation}
\label{eq:theta-S}
\sum_{j=0}^{\lfloor n/8\rfloor}
\frac{(-1)^j}{16^j} a_j\theta_2(q)^{n-8j}\theta_4(q^2)^{8j},
\end{equation}
where
$\Delta_8(q) = q \prod_{m=1}^{\infty} (1 - q^{2m-1})^8(1-q^{4m})^8$,
and $\theta_2(q), \theta_3(q)$, $\theta_4(q)$ are the Jacobi
theta series~\cite{SPLAG}.
%% As the additional conditions, it holds that
%% there is at most one nonzero $B_r$ for $r < (\mu+2)/2$;
%% $B_r=0$ for $r < \mu/4$; and $B_r \le 2$ for $r < \mu/2$,
%% where $\mu$ is the minimum norm of $L$.

%%%%%%%%%%%%%%%%%%%%%%%%%%%%%
\subsection{Extremal unimodular lattices and $s$-extremal unimodular lattices}

By considering the theta series of odd unimodular lattices and their shadows,
Rains and Sloane~\cite{RS-bound} showed that
the minimum norm $\min(L)$ of an odd unimodular
lattice $L$ in dimension $n$
is bounded by
\begin{equation}\label{eq:RS}
\min(L) \le 
\begin{cases}
2 \lfloor \frac{n}{24} \rfloor+2 &\text{ if } n \ne 23, \\
3 &\text{ if } n = 23.
\end{cases}
\end{equation}
A unimodular lattice meeting the bound~\eqref{eq:RS} with equality is called 
\emph{extremal}.
In addition, it was shown in~\cite{Ga07} that 
\begin{equation}\label{eq:B}
%\begin{split}
%4\min(S(L))=15 & \text{ if } n=23 \text{ and }\min(L)=3, \\ 
%8 \min(L) + 4 \min(S(L)) \le 8+n &\text{ otherwise},
%\end{split}
\begin{array}{rcll}
4\min(S(L))&=&15 & \text{ if } n=23 \text{ and }\min(L)=3, \\ 
8 \min(L) + 4 \min(S(L)) &\le& 8+n &\text{ otherwise},
\end{array}
\end{equation}
where $\min(S(L))$ 
denotes the minimum norm of $S(L)$, that is,
the smallest norm among all vectors of $S(L)$.
An odd unimodular lattice meeting the bound~\eqref{eq:B} with equality is called 
\emph{$s$-extremal}.
% We note that $s$-extremal unimodular lattices 
% have been widely studied for minimum norms $\mu=2,3,4$
% (see e.g.\~\cite{Ga07},~\cite{H11} and the references therein).
If an $s$-extremal unimodular lattice $L$ in dimension $n$ having
$\min(L)=4$ exists, then 
$n \in \{32,36,37,\ldots,47\}$~\cite[p.~148]{Ga07}.

\subsection{Invariants $\inv(L)_{i}$}\label{sec:inv}
For a given unimodular lattice $L$ in dimension $n$ having $\min(L)=4$,
let $L(4)$ denote the set of vectors of norm $4$ in $L$.
There is a subset $L(4)^+$ of $L(4)$ such that
\[
L(4)=L(4)^+ \cup L(4)^- \text{ and }
L(4)^+ \cap L(4)^- = \emptyset,
\]
where 
$L(4)^-=\{-x \mid x \in L(4)^+\}$.
For a given unimodular lattice $L$ and a nonnegative integer $i$,
we define
\[
\inv(L)_{i}=|\{ (x,y) \in L(4)^+ \times L(4)^+ \mid x \cdot y \in \{i,-i\}\}|.
\]
It is trivial that  $\inv(L)_{i}=\inv(L')_{i}$ if $L \cong L'$
for a nonnegative integer $i$.

\begin{lem}\label{lem:L4}
$\inv(L)_{3}=0$, $\inv(L)_{4}=|L^+(4)|$ and 
$\inv(L)_{0}+\inv(L)_{1}+\inv(L)_{2}=|L^+(4)|^2-|L^+(4)|$.
\end{lem}
\begin{proof}
Let $x$ and $y$ be vectors of $L(4)$.
By considering $(x+y) \cdot(x+y)$ and $(x-y) \cdot(x-y)$,
it follows that 
$-2 \le x \cdot y \le 2$ if $x \not \in \{y,-y\}$, and 
it follows that 
$x \in \{y,-y\}$ if and only if $x \cdot y \in \{-4,4\}$.
The last assertion is trivial.
\end{proof}

Hence, it is sufficient to consider only
$\inv(L)_{0}$ and $\inv(L)_{1}$. 

%%%%%%%%%%%
\subsection{Self-dual codes and Construction A}
Let $\FF_5$ be the finite field of order $5$.
Throughout this paper, 
we take the elements of $\FF_5$ to be either
$0,1,2,3,4$ or $0,\pm1, \pm 2$,
using whichever form is more convenient.
A (linear) $[n,k]$ \emph{code} $C$ over $\FF_5$ is a $k$-dimensional subspace of 
$\FF_5^n$.
All codes in this paper mean linear codes 
and we omit the term linear.
% The value $n$ is called the \emph{length} of $C$.
The \emph{weight} $\wt(x)$ of a vector $x$ of $\FF_5^n$ is 
the number of non-zero components of $x$.
A vector of $C$ is called a \emph{codeword}.
The minimum non-zero weight of all codewords in $C$ is called
the \emph{minimum weight} of $C$.
An $[n,k]$ code with minimum weight $d$ is called an $[n,k,d]$ code.
The \emph{weight enumerator} of $C$ is $\sum_{c \in C} y^{\wt(c)}$.

The \emph{dual} code $C^{\perp}$ of an $[n,k]$ code $C$ over $\FF_5$ is defined as
\[
C^{\perp}=
\{x \in \FF_5^n \mid x \cdot y = 0 \text{ for all } y \in C\},
\]
where $x \cdot y$ is the standard inner product.
A code $C$ is  called \emph{self-dual} if $C = C^{\perp}$. 
If an $[n,k]$ code is self-dual, then it is trivial that $k=n/2$.
% For $p \equiv 1 \pmod 4$, a self-dual $[n,n/2]$ code over $\FF_p$ exists
% if and only if $n$ is even, and for $p \equiv 3 \pmod 4$, 
% a self-dual  $[n,n/2]$
% code over $\FF_p$ exists if and only if $n \equiv 0 \pmod4$.
%%%%%%%%%%%%%%%% 
%% Two codes $C$  nd $C'$ are \textit{equivalent} if there is some monomial
%% matrix $M$ over $\FF_q$ such that $C' =C M =\{c M | c \in C \}$.
%% A monomial matrix which maps $C$ to itself is called an automorphism 
%% of $C$ and  the set of all automorphisms of $C$ forms the 
%% automorphism group $\Aut(C)$ of $C$.
%% It is a fundamental problem to classify self-dual codes 
%% over $\FF_q$ and determine the largest minimum weight 
%% among self-dual codes for a fixed length $n$.
It is a fundamental problem to determine the largest minimum weight 
$d_5(n)$ among all self-dual  $[n,n/2]$ codes over $\FF_5$ for
a given $n$.

Let $C$ be a self-dual $[n,n/2]$ code over $\FF_5$.
Then the following lattice
\[
A_{5}(C) = \frac{1}{\sqrt{5}}
\{(x_1,\ldots,x_n) \in \ZZ^n \mid
(x_1 \!\!\!\pmod 5,\ldots,x_n \!\!\!\pmod 5)\in C\}
\]
is a unimodular lattice in dimension $n$.
This construction of lattices is well known as \emph{Construction~A}.

%%%%%%%%%%%%%%%%%%%%%%%%%%%%%
\subsection{Self-dual codes constructed from negacirculant matrices}

An $n \times n$
\emph{negacirculant} matrix has the following form
\[
\left( \begin{array}{ccccc}
r_0&r_1&r_2& \cdots &r_{n-1} \\
-r_{n-1}&r_0&r_1& \cdots &r_{n-2} \\
-r_{n-2}&-r_{n-1}&r_0& \cdots &r_{n-3} \\
\vdots &\vdots & \vdots && \vdots\\
-r_1&-r_2&-r_3& \cdots&r_0
\end{array}
\right).
\]
Throughout this paper, 
let $I_n$ denote the identity matrix of order $n$
and 
let $A^T$ denote the transpose of a matrix $A$.

Let $A$ be an $n \times n$ negacirculant matrix.
A $[2n,n]$ code over
$\FF_5$ having the following generator matrix
\begin{equation} \label{eq:2}
\left(
\begin{array}{cc}
I_{n} & A
\end{array}
\right)
\end{equation}
% If  $AA^T+BB^T=-I_n$, then $C$ is self-dual~\cite{HHKK}.
is called a \emph{quasi-twisted} code
or a double twistulant code.
%For an $n \times n$ matrix $A$, $A^T$ denotes the transposed matrix of $A$.
Let $A$ and $B$ be $n \times n$ negacirculant matrices.
A $[4n,2n]$ code over
$\FF_5$ having the following generator matrix
\begin{equation} \label{eq:4}
\left(
\begin{array}{ccc@{}c}
\quad & {\Large I_{2n}} & \quad &
\begin{array}{cc}
A & B \\
-B^T & A^T
\end{array}
\end{array}
\right)
\end{equation}
% If  $AA^T+BB^T=-I_n$, then $C$ is self-dual~\cite{HHKK}.
is called a \emph{four-negacirculant} code.
Many quasi-twisted self-dual codes and four-negacirculant 
self-dual codes with large minimum weights are known 
(see e.g.~\cite{GH08}, \cite{GHM}, \cite{H15} and \cite{HHKK}).

%%%%%%%%%%%%%%%%%%%%%%%%%%%%%%%%%
\section{New $s$-extremal extremal unimodular lattices in dimension $44$
and new self-dual codes over $\FF_5$}
\label{sec:44}

In this section,
new $s$-extremal extremal unimodular lattices in dimension $44$
are constructed from four-negacirculant self-dual codes over $\FF_5$ 
by Construction~A\@.
In the process of constructing these codes, 
we obtain a self-dual $[44,22,14]$ code over $\FF_5$.
This code has larger minimum weight than the 
previously known $[44,22]$ codes.
In addition, the new self-dual code implies
a $[43,22,13]$ code over $\FF_5$, which has larger minimum weight than the 
previously known $[43,22]$ codes.

%%%%%%%%%%%%
\subsection{New $s$-extremal extremal unimodular lattices in dimension $44$}

%Recall that
%the \emph{theta series} of an odd unimodular lattice $L$ and its shadow $S(L)$
%are the formal power series
%$\theta_{L}(q) = \sum_{x \in L} q^{x \cdot x}$ and
%$\theta_{S(L)}(q) = \sum_{x \in S(L)} q^{x \cdot x}$, respectively.
Let $L_{44}$ be an odd unimodular lattice 
in dimension $44$ having minimum norm $4$.
By using~\eqref{eq:theta} and~\eqref{eq:theta-S},
the possible theta series 
of $L_{44}$ and its shadow $S(L_{44})$ are determined as follows:
\begin{align*}
\theta_{L_{44}}(q) &=
1 
+ (6600 + 16 \alpha) q^4 
+ (811008  - 128 \alpha- 65536 \beta) q^5 
%% \\ & \hspace*{3cm}
%% + (37171200 - 128 \alpha + 2097152 \beta) q^6 
+ \cdots \text{ and}
\\
\theta_{S(L_{44})}(q) &=
\beta q
+(\alpha - 76 \beta) q^3
+ (1622016 - 52 \alpha + 2806 \beta) q^5 
+ \cdots,
\end{align*}
respectively,
where $\alpha$ and $\beta$ are nonnegative integers~\cite{H11}.
It turns out that $L_{44}$ has kissing number $6600$ 
if and only if $L_{44}$ is $s$-extremal.
Two $s$-extremal extremal unimodular lattices
in dimension $44$ are previously known.
More precisely,
$A_5(C_{44})$ in~\cite{H11} and 
$A_5(C_{44,5}(D_{22}))$ in~\cite{H15} are the known lattices, and we
denote the lattices by $L_{44,1}$ and $L_{44,2}$, respectively.
We remark that the two lattices are constructed from some self-dual codes
over $\FF_5$ by Construction~A\@.

We calculated by \textsc{Magma} the invariants $\inv(L)_{i}$ 
$(i=0,1)$ given in Section~\ref{sec:inv}
for $L=L_{44,1}$ and $L_{44,2}$,
where the results are listed in Table~\ref{Tab:inv}.
This was done using the \textsc{Magma} function \texttt{ShortVectors}.
%From the table, we have the following:
From the table, it holds that
the two lattices $L_{44,1}$ and $L_{44,2}$ are non-isomorphic.

% \begin{lem}
% The two lattices $L_{44,1}$ and $L_{44,2}$
% % $A_5(C_{44})$ in~\cite{H11} and $A_5(C_{44,5}(D_{22}))$ in~\cite{H15}
% are non-isomorphic.
% \end{lem}

       %%%%%%%%%%%%%%%%%  table  %%%%%%%%%%%%%%%%%
\begin{table}[th]
\caption{$\inv(L)_{i}$ $(i=0,1)$ for $L=L_{44,1}$ and $L_{44,2}$}
\label{Tab:inv}
\centering
\medskip
{\small
%{\footnotesize
%{\scriptsize
\begin{tabular}{c|cc}
\noalign{\hrule height1pt}
$L$ & $\inv(L)_{0}$ & $\inv(L)_{1}$ \\
\hline
% $A_5(C_{5,44})$&7107804&3736128 &$K_1$&7097112&3750384\\
% $A_7(C_{7,44})$&7088400&3762000 &$L_{44,2}$&7089192&3760944\\
$L_{44,1}$&7097112&3750384\\
$L_{44,2}$&7089192&3760944\\
%%$A_5(C_{5,44})$&7107804&3736128 \\
\noalign{\hrule height1pt}
\end{tabular}
}
\end{table}
     %%%%%%%%%%%%%%%%%  table  %%%%%%%%%%%%%%%%%

In order to construct new $s$-extremal extremal odd unimodular lattices,
we tried to find four-negacirculant
self-dual $[44,22]$ codes over $\FF_5$.
Then we found $50$ codes $C_{44,i}$ satisfying the condition that
$A_5(C_{44,i})$ have minimum norm $4$ and kissing number $6600$ and 
%the condition that $\inv(L)_{0}$ or $\inv(L)_{1}$ is distinct for 
%$L=L_{44,1}$, $L_{44,2}$ and $A_5(C_{44,i})$ $(i=1,2,\ldots,50)$.
the condition that
\[
|\{(\inv(L)_{0},\inv(L)_{1}) \mid L \in \mathcal{L}\}| =52,
\]
where
$\mathcal{L}=\{L_{44,1}, L_{44,2}\} \cup 
\{A_5(C_{44,i}) \mid i \in \{1,2,\ldots,50\}\}$.
The self-duality was verified by
the \textsc{Magma} function \texttt{IsSelfDual}.
The minimum norm and the kissing number
were calculated by
the \textsc{Magma} functions \texttt{Minimum} and \texttt{KissingNumber},
respectively.
For $L=A_5(C_{44,i})$,
the results $\inv(L)_{j}$ $(j=0,1)$
are listed in Table~\ref{Tab:inv442}.
From Tables~\ref{Tab:inv} and \ref{Tab:inv442},
we have the following:

\begin{lem}
The $52$ lattices
$L_{44,1}$, $L_{44,2}$ and 
$A_5(C_{44,i})$ $(i=1,2,\ldots,50)$
are non-isomorphic.
\end{lem}

       %%%%%%%%%%%%%%%%%  table  %%%%%%%%%%%%%%%%%
\begin{table}[thb]
\caption{$\inv(L)_{j}$ $(j=0,1)$ for $L=A_5(C_{44,i})$}
\label{Tab:inv442}
\centering\medskip
%{\small
{\footnotesize
%{\scriptsize
\begin{tabular}{c|cc||c|cc||c|cc}
\noalign{\hrule height1pt}
$i$ & $\inv(L)_{0}$ & $\inv(L)_{1}$ & 
$i$ & $\inv(L)_{0}$ & $\inv(L)_{1}$ & 
$i$ & $\inv(L)_{0}$ & $\inv(L)_{1}$ \\
\hline
 1&7068600& 3788400 &18&7083252& 3768864 &35&7092756& 3756192 \\
 2&7071372& 3784704 &19&7083648& 3768336 &36&7093152& 3755664 \\
 3&7074144& 3781008 &20&7084440& 3767280 &37&7093548& 3755136 \\
 4&7074540& 3780480 &21&7085232& 3766224 &38&7093944& 3754608 \\
 5&7074936& 3779952 &22&7085628& 3765696 &39&7094340& 3754080 \\
 6&7075332& 3779424 &23&7086420& 3764640 &40&7094736& 3753552 \\
 7&7076520& 3777840 &24&7086816& 3764112 &41&7095132& 3753024 \\
 8&7077312& 3776784 &25&7087212& 3763584 &42&7095528& 3752496 \\
 9&7078104& 3775728 &26&7087608& 3763056 &43&7096320& 3751440 \\
10&7078500& 3775200 &27&7088004& 3762528 &44&7097904& 3749328 \\
11&7078896& 3774672 &28&7089588& 3760416 &45&7098300& 3748800 \\
12&7079292& 3774144 &29&7089984& 3759888 &46&7100676& 3745632 \\
13&7079688& 3773616 &30&7090380& 3759360 &47&7103448& 3741936 \\
14&7081272& 3771504 &31&7091172& 3758304 &48&7105824& 3738768 \\
15&7081668& 3770976 &32&7091568& 3757776 &49&7107012& 3737184 \\
16&7082460& 3769920 &33&7091964& 3757248 &50&7107804& 3736128 \\
17&7082856& 3769392 &34&7092360& 3756720 &  &       &         \\
\noalign{\hrule height1pt}
\end{tabular}
}
\end{table}
%% 50&7118100& 3722400 \\ old data
     %%%%%%%%%%%%%%%%%  table  %%%%%%%%%%%%%%%%%

For the $50$ codes $C_{44,i}$, 
the first rows $r_A(C_{44,i})$ and $r_B(C_{44,i})$
of the negacirculant matrices $A$ and $B$ in
the generator matrices of form~\eqref{eq:4} are listed in Table~\ref{Tab:44}.

For an odd unimodular lattice $L$, 
there are cosets $L_1,L_2,L_3$ of $L_0$ such that
$L_0^* = L_0 \cup L_1 \cup L_2 \cup L_3$, where
$L = L_0  \cup L_2$ and $S(L) = L_1 \cup L_3$.
Two lattices $L$ and $L'$ are \emph{neighbors} if
both lattices contain a sublattice of index $2$
in common.
If $L$ is an odd unimodular lattice in dimension $n$ and $n$ is a multiple of four, 
then there are two unimodular lattices
containing $L_0$,
which are rather than $L$,
namely, $L_0 \cup L_1$ and $L_0 \cup L_3$ (see~\cite{DHS}).
Note that the two neighbors are even if $n$ is a multiple of eight.
We denote the two unimodular neighbors by
% \begin{equation}\label{eq:N}
\[
N_1(L)=L_0 \cup L_1 \text{ and } N_2(L)=L_0 \cup L_3.
\]
% \end{equation}

\begin{lem}
Let $L$ be an $s$-extremal extremal unimodular lattice in dimension $44$.
Then $N_1(L)$ and $N_2(L)$ are also 
$s$-extremal extremal unimodular lattices.
\end{lem}
\begin{proof}
Since $L$ is an $s$-extremal extremal unimodular lattice in dimension $44$,
the minimum norm of the shadow $S(L)$ is $5$.
Thus, $N_1(L)$ and $N_2(L)$ have minimum norm $4$.
In addition, the shadows of $N_1(L)$ and $N_2(L)$ 
are $L_2 \cup L_3$ and $L_2 \cup L_1$, respectively.
Hence, $\min(S(N_1(L))=\min(S(N_2(L))=5$.
The result follows.
\end{proof}

By the above lemma, 
more $s$-extremal extremal unimodular lattices in dimension $44$
are constructed as 
$N_j(L_{44,1})$, $N_j(L_{44,2})$ and
$N_j(A_5(C_{44,i}))$ $(i=1,2,\ldots,50)$
$(j=1,2)$. 

\begin{lem}\label{lem:inv}
Let $L$ and $L'$ be $s$-extremal extremal unimodular lattices in dimension $44$.
% If $L$ and $L'$ have distinct $\inv(L)_{j}$ $(j=0,1)$, 
If $\inv(L)_{0}\ne \inv(L')_{0}$ or $\inv(L)_{1}\ne \inv(L')_{1}$,
then $L \not \cong N_1(L')$ and $L \not \cong N_2(L')$.
\end{lem}
\begin{proof}
%From the theta series of 
%$s$-extremal extremal unimodular lattices in dimension $44$ and its shadow,
Since the minimum norm of the shadow of $L'$ is $5$,
the two sets of vectors of norm $4$ in $L'$ and $M$ are identical
for $M=N_1(L')$ and $N_2(L')$.
%the three lattices $L',N_1(L')$ and $N_2(L')$ have the identical
%sets of vectors of norm $4$.
The result follows.
\end{proof}

In addition, we verified that the three lattices
$L,N_1(L)$ and $N_2(L)$ are non-isomorphic
for each lattice $L=L_{44,1}$, $L_{44,2}$
and
$A_5(C_{44,i})$ $(i=1,2,\ldots,50)$.
The neighbors $N_1(L)$ and $N_2(L)$ were constructed using
the \textsc{Magma} functions \texttt{EvenSublattice} and \texttt{Dual}.
In addition, 
the non-isomorphisms were verified by the \textsc{Magma}
function \texttt{IsIsomorphic}.
By Lemma~\ref{lem:inv}, we have the following:

\begin{prop}\label{prop:L}
There are at least $156$ non-isomorphic 
$s$-extremal extremal unimodular lattices in dimension $44$.
\end{prop}

We stopped our search after finding the $50$ self-dual codes $C_{44,i}$, which give
$150$ lattices
$A_5(C_{44,i})$, $N_1(A_5(C_{44,i}))$ and $N_2(A_5(C_{44,i}))$.
Our feeling is that the number of non-isomorphic 
$s$-extremal extremal unimodular lattices in dimension $44$
might be even bigger.

% We verified by \textsc{Magma} that
% $A_5(C_{44})$, $N_1(A_5(C_{44}))$ and $N_2(A_5(C_{44}))$
% have automorphism groups of order $88$, and 
% $L,N_1(L)$ and $N_2(L)$
% have automorphism groups of order $44$
% for $L=A_7(C_{7,44})$, $K_1$ and $L_{44,2}$.

%%%%%%%%%%%%%%%%%%%%%%%%%%%%%%%%%%%%%%%%%%%%%%%%
\subsection{A self-dual $[44,22,14]$ code over $\FF_5$ and its related codes}

We verified by \textsc{Magma} that 
$C_{44,50}$ has minimum weight $14$,
$C_{44,29}$ has minimum weight $13$ and 
$C_{44,i}$ $(i=1,2,\ldots,28,30,\ldots,49)$ have minimum weight $12$.
The minimum weight was calculated by
the \textsc{Magma} function \texttt{MinimumWeight}.

\begin{prop}
There is a self-dual $[44,22,14]$ code over $\FF_5$.
\end{prop}

The code $C_{44,50}$ has generator matrix of form~\eqref{eq:4},
where the negacirculant matrices $A$ and $B$ have the first rows
\[
(10033210404) \text{ and }
(12241413344),
\]
respectively.
The first few terms of the weight enumerator of $C_{44,50}$ are given by
\[
1 + 12056 y^{14} + 95920 y^{15} + 807312 y^{16} + 4677728 y^{17} + \cdots. 
\]
This was calculated by the \textsc{Magma} function \texttt{NumberOfWord}.

Let $d_5(n)$ denote the largest minimum weight 
among all self-dual  $[n,n/2]$ codes over $\FF_5$.
It was known that $13 \le d_5(44) \le 19$~\cite[Table~III]{GH08}.
Hence, $C_{44,50}$ improves 
the previously known lower bound on the largest minimum weight $d_5(44)$.
As a summary,  we list the current information on
the largest minimum weight $d_5(n)$ in Table~\ref{Tab:F5}
along with references for $22 \le n \le 72$.
The table updates~\cite[Table~III]{GH08}.

%%%%%%%%%%%%%%%%%%%%%%%%%%%%%%%%%%%%%%%%%%%%%%%%
\begin{table}[thb]
%% \caption{Summary}
\caption{Largest minimum weights $d_5(n)$}
\label{Tab:F5}
\centering
\medskip
{\small
%{\footnotesize
%{\scriptsize
\begin{tabular}{c|c|c||c|c|c}
\noalign{\hrule height0.8pt}
$n$ & $d_5(n)$ & References & $n$ & $d_5(n)$ & References \\ 
\hline
22 &  8--10 &\cite{LPS-GF5}   & 48 & 14--20 &\cite{G02}    \\
24 &  9--10 &\cite{LPS-GF5}   & 50 & 14--20 &\cite{GaOt}    \\
26 &  9--11 &\cite{GaOt}      & 52 & 15--21 &\cite{GaOt}    \\
28 & 10--12 &\cite{GaOt}      & 54 & 16--22 &\cite{GaOt}    \\
30 & 10--12 &\cite{GaOt}      & 56 & 16--23 &\cite{GaOt}    \\
32 & 11--13 &\cite{GHM}       & 58 & 16--24 &\cite{G02}    \\
34 & 11--14 &\cite{GaOt}      & 60 & 18--24 &\cite{GH} \\%(see~\cite{G02})    \\
36 & 12--15 &\cite{GaOt}      & 62 & 17--25 &\cite{GaOt}    \\
38 & 12--16 &\cite{G02}      & 64 & 18--26 &\cite{GH} \\%(see~\cite{G02})    \\
40 & 13--17 &\cite{G02}      & 66 & 18--27 &\cite{GG}    \\
42 & 13--18 &\cite{GH08}      & 68 & 18--28 &\cite{GaOt}    \\
44 & 14--19 & $C_{44,50}$       & 70 & 20--29 &\cite{GG}    \\
46 & 14--20 &\cite{G02}      & 72 & 22--29 &\cite{GG}    \\
\noalign{\hrule height0.8pt}
\end{tabular}
}
\end{table}

It is a main problem in coding theory to determine
the largest minimum weights $d_q(n,k)$ among
all $[n,k]$ codes over a finite field of order $q$
for a given $(q,n,k)$.
The current information on $d_5(n,k)$ can be found in~\cite{Grassl}.
For example, it was known that 
$12 \le d_5(43,22) \le 18$ and
$13 \le d_5(44,22) \le 19$.
We emphasize that $C_{44,50}$ is the first example as not only
self-dual $[44,22,14]$ codes but also (linear) $[44,22,14]$ codes.
We verified that all punctured codes of $C_{44,50}$ are $[43,22,13]$ codes over $\FF_5$.
The punctured codes were constructed by
the \textsc{Magma} functions \texttt{PunctureCode}.

\begin{prop}
There is a $[43,22,13]$ code over $\FF_5$.
\end{prop}

The self-dual $[44,22,14]$ code $C_{44,50}$ and 
the punctured codes
improve the previously known lower bounds on
$d_5(43,22)$ and $d_5(44,22)$.

\begin{cor}
$13 \le d_5(43,22) \le 18$ and
$14 \le d_5(44,22) \le 19$. 
\end{cor}

%%%%%%%%%%%%%%%%%%%%%%%%%%%%%%%%%
\section{Construction of $s$-extremal extremal unimodular lattices
in dimensions $38,40$ and $42$}\label{sec:lat}

In this section, we construct
new $s$-extremal extremal unimodular lattices in dimensions $38,40$ and $42$
from self-dual codes over $\FF_5$ 
by Construction~A\@.
These self-dual codes are constructed as
four-negacirculant self-dual codes or quasi-twisted self-dual codes.

%%%%%%%%%%%%
\subsection{New $s$-extremal extremal odd unimodular lattices in dimension $40$}

By using~\eqref{eq:theta} and~\eqref{eq:theta-S},
the possible theta series of an extremal odd unimodular lattice $L_{40}$ 
in dimension $40$ and
its shadow $S(L_{40})$ are determined as follows:
\begin{align*}
\theta_{L_{40}}(q) &=
1
+ (19120  + 256 \alpha) q^4
+ (1376256 - 4096 \alpha) q^5 + \cdots \text{ and}
\\
\theta_{S(L_{40})}(q) &=
\alpha q^2
+ (40960 - 56 \alpha) q^4
+ (87818240 + 1500 \alpha) q^6 +\cdots,
\end{align*}
respectively,
where $\alpha$ is an even integer with $0 \le \alpha \le 80$~\cite{BBH}.
It is trivial  that $L_{40}$ has kissing number $19120$ 
if and only if $L_{40}$ is $s$-extremal.
An $s$-extremal extremal odd unimodular lattice in dimension $40$ was
explicitly constructed in~\cite{BBH} and 
three non-isomorphic $s$-extremal extremal odd unimodular lattices
were explicitly constructed in~\cite{H15}.
We calculated by \textsc{Magma} $\inv(L)_{i}$ 
$(i=0,1)$ for the four lattices, where the results are
listed in Table~\ref{Tab:inv40k}.

       %%%%%%%%%%%%%%%%%  table  %%%%%%%%%%%%%%%%%
\begin{table}[th]
  \caption{$\inv(L)_{i}$ $(i=0,1)$ for the known lattices $L$
    in~\cite{BBH} and~\cite{H15}}
  \label{Tab:inv40k}
\centering\medskip
%{\small
{\footnotesize
%{\scriptsize
\begin{tabular}{c|cc}
\noalign{\hrule height1pt}
$L$ & $\inv(L)_{0}$ & $\inv(L)_{1}$ \\
\hline
\cite{BBH}                         & 56589480 & 34257920\\
$A_{13}(C_{13,40})$ in~\cite{H15}  & 56644200 & 34184960\\
$A_9(C_{9,40})$ in~\cite{H15}      & 56549160 & 34311680\\
$A_{19}(C_{19,40})$ in~\cite{H15}  & 56553480 & 34305920\\
\noalign{\hrule height1pt}
\end{tabular}
}
\end{table}
     %%%%%%%%%%%%%%%%%  table  %%%%%%%%%%%%%%%%%

In order to construct new $s$-extremal extremal odd unimodular lattices,
we tried to find self-dual $[40,20]$ codes $C$ over $\FF_5$ by considering 
four-negacirculant codes such that $A_5(C)$ have minimum norm $4$ and
kissing number $19120$.
Then we found $50$ self-dual $[40,20]$ codes $C_{40,i}$ over $\FF_5$ 
such that their lattices $A_5(C_{40,i})$ are
$s$-extremal extremal odd unimodular lattices and 
%$\inv(A_5(C_{40,i}))_{0}$ or $\inv(A_5(C_{40,i}))_{1}$ is distinct.
\[
|\{(\inv(L)_{0},\inv(L)_{1}) \mid L \in \{A_5(C_{40,i}) \mid i \in \{1,2,\ldots,50\}\}|
 =50.
\]
The results $\inv(A_5(C_{40,i}))_{j}$ $(j=0,1)$
are listed in Table~\ref{Tab:inv40}.
%Since the $50$ lattices $A_5(C_{40,i})$ have distinct
%$\inv(A_5(C_{40,i}))_{4,j}$ $(j=0,1,2,3,4)$,
%the lattices are non-isomorphic.
From Tables~\ref{Tab:inv40k} and~\ref{Tab:inv40}, we have the following:

\begin{prop}\label{prop:L40}
There are at least $54$ non-isomorphic 
$s$-extremal extremal odd unimodular lattices in dimension $40$.
\end{prop}

       %%%%%%%%%%%%%%%%%  table  %%%%%%%%%%%%%%%%%
\begin{table}[thb]
\caption{$\inv(L)_{j}$ $(j=0,1)$ for $L=A_5(C_{40,i})$}
\label{Tab:inv40}
\centering\medskip
%{\small
{\footnotesize
%{\scriptsize
\begin{tabular}{c|cc||c|cc||c|cc}
\noalign{\hrule height1pt}
$i$ & $\inv(L)_{0}$ & $\inv(L)_{1}$ & 
$i$ & $\inv(L)_{0}$ & $\inv(L)_{1}$ & 
$i$ & $\inv(L)_{0}$ & $\inv(L)_{1}$ \\
\hline
 1 &56523240& 34346240&18 &56593800& 34252160&35 &56635560& 34196480\\
 2 &56536200& 34328960&19 &56596680& 34248320&36 &56637000& 34194560\\
 3 &56554920& 34304000&20 &56598120& 34246400&37 &56648520& 34179200\\
 4 &56559240& 34298240&21 &56601000& 34242560&38 &56652840& 34173440\\
 5 &56562120& 34294400&22 &56605320& 34236800&39 &56654280& 34171520\\
 6 &56563560& 34292480&23 &56606760& 34234880&40 &56657160& 34167680\\
 7 &56566440& 34288640&24 &56609640& 34231040&41 &56661480& 34161920\\
 8 &56570760& 34282880&25 &56611080& 34229120&42 &56667240& 34154240\\
 9 &56572200& 34280960&26 &56613960& 34225280&43 &56670120& 34150400\\
10 &56575080& 34277120&27 &56618280& 34219520&44 &56674440& 34144640\\
11 &56576520& 34275200&28 &56619720& 34217600&45 &56678760& 34138880\\
12 &56579400& 34271360&29 &56622600& 34213760&46 &56704680& 34104320\\
13 &56580840& 34269440&30 &56624040& 34211840&47 &56734920& 34064000\\
14 &56583720& 34265600&31 &56626920& 34208000&48 &56753640& 34039040\\
15 &56585160& 34263680&32 &56628360& 34206080&49 &56825640& 33943040\\
16 &56588040& 34259840&33 &56631240& 34202240&50 &56851560& 33908480\\
17 &56592360& 34254080&34 &56632680& 34200320&   &        &         \\
\noalign{\hrule height1pt}
\end{tabular}
}
\end{table}
     %%%%%%%%%%%%%%%%%  table  %%%%%%%%%%%%%%%%%

% We verified by \textsc{Magma} that
% \[
% |\Aut(A_5(C_{40,i}))| =
% \begin{cases}
% 80 & \text{ if } i \in \{1, 25, 37\}, \\
% 40 & \text{ otherwise.} 
% \end{cases}
% \]
For the $50$ codes $C_{40,i}$, 
the first rows $r_A(C_{40,i})$ and $r_B(C_{40,i})$
of the negacirculant matrices $A$ and $B$ in
the generator matrices of form~\eqref{eq:4} are listed in Table~\ref{Tab:40}.
We verified by \textsc{Magma}
that the $50$ codes $C_{40,i}$ have minimum weight $12$.
Although these codes have minimum weights less than $d_5(40)$,
these codes are useful for constructing $s$-extremal extremal odd unimodular
lattices.

We stopped our search after finding the new $50$ lattices.
Our feeling is that the number of non-isomorphic 
$s$-extremal extremal unimodular lattices in dimension $40$
might be even bigger.

%%%%%%%%%%%%
\subsection{New $s$-extremal extremal odd unimodular lattices in dimension $38$}

Let $L_{38}$ be an extremal odd unimodular lattice in dimension $38$.
By using~\eqref{eq:theta} and~\eqref{eq:theta-S},
one can determine the possible theta series $\theta_{L_{38}}(q)$ and
$\theta_{S(L_{38})}(q)$ as follows.
Since the minimum norm of $L_{38}$ is $4$, 
we have that
\[
a_0=1,
a_1=-76,
a_2=1140 \text{ and }
a_3=-1520,
\]
in~\eqref{eq:theta}.
By considering the coefficient of $q^{\frac{3}{2}}$ in~\eqref{eq:theta-S},
$a_4$ is written as:
\[
a_4=2^{10} \alpha,
\]
by using an integer $\alpha$.
Hence, we have that
\begin{equation*}\label{eq:theta38}
\begin{split}
\theta_{L_{38}}(q) &=
1
+( 29260  + 1024 \alpha )q^4
+( 1668352  - 20480 \alpha) q^5
 + \cdots \text{ and} \\
%\\& +( 44304960 +  176128 \alpha) q^6 
 \theta_{S(L_{38})}(q) &=
\alpha q^{\frac{3}{2}}
+( 6080 - 58 \alpha) q^{\frac{7}{2}}
+( 18471040 + 1615 \alpha) q^{\frac{11}{2}} 
+ \cdots.
\end{split}
\end{equation*}
It is trivial that 
$L_{38}$ has kissing number $29260$ 
if and only if $L_{38}$ is $s$-extremal.
An $s$-extremal extremal unimodular lattice
in dimension $38$ is previously known and 
this lattice is denoted by $G_{38}$ in~\cite{Ga04}.
We calculated by \textsc{Magma}
\begin{equation}\label{eq:G38}
\inv(G_{38})_{0}=129060350 \text{ and }
\inv(G_{38})_{1}= 83320320.
\end{equation}

       %%%%%%%%%%%%%%%%%  table  %%%%%%%%%%%%%%%%%
\begin{table}[th]
\caption{$r(C)$ and $\inv(A_5(C))_{j}$ $(j=0,1)$ for $C=C_{38,i}$}
\label{Tab:inv38}
\centering
\medskip
%{\small
{\footnotesize
%{\scriptsize
\begin{tabular}{c|c|cc}
\noalign{\hrule height1pt}
$i$ & $r(C)$ & $\inv(A_5(C))_{0}$& $\inv(A_5(C))_{1}$\\
\hline
$1$& $(1,4,0,1,2,3,3,1,3,4,0,4,4,0,1,2,1,1,2)$& 128961854& 83451648\\ 
$2$& $(1,0,0,0,4,2,2,2,1,0,3,4,2,0,3,1,2,1,0)$& 129027518& 83364096\\ 
$3$& $(1,0,4,3,0,0,3,2,1,3,3,0,1,3,4,0,0,3,4)$& 129060350& 83320320\\ 
$4$& $(1,0,0,2,1,1,1,1,1,1,0,0,2,0,2,2,0,0,1)$& 129093182& 83276544\\ 
$5$& $(1,0,0,2,2,0,4,0,0,2,3,1,0,0,2,1,0,1,3)$& 129126014& 83232768\\ 
$6$& $(1,0,0,3,3,3,4,3,1,0,3,2,0,4,4,4,4,2,3)$& 129126014& 83232768\\ 
$7$& $(1,0,2,1,1,3,3,2,1,0,4,2,3,0,0,3,1,4,2)$& 129126014& 83232768\\ 
$8$& $(1,0,0,4,0,0,2,3,4,2,1,0,0,3,1,4,1,0,4)$& 129158846& 83188992\\ 
$9$& $(1,0,3,2,0,1,1,4,3,1,0,1,3,4,1,2,4,3,4)$& 129158846& 83188992\\ 
$10$& $(1,0,3,2,2,2,3,2,0,1,3,3,3,0,4,2,1,4,3)$& 129191678& 83145216\\ 
$11$& $(1,0,0,2,1,2,0,0,2,1,0,3,3,0,2,1,4,2,1)$& 129224510& 83101440\\ 
$12$& $(1,0,3,3,3,0,2,1,1,3,2,4,1,3,1,3,3,4,1)$& 129224510& 83101440\\ 
$13$& $(1,0,0,1,3,2,3,0,2,1,1,4,3,1,1,3,1,1,1)$& 129257342& 83057664\\ 
$14$& $(1,0,1,4,4,4,4,1,2,0,3,3,4,2,2,0,4,1,3)$& 129257342& 83057664\\ 
$15$& $(1,0,2,2,3,2,4,2,3,3,1,2,4,4,1,0,0,1,0)$& 129257342& 83057664\\ 
\noalign{\hrule height1pt}
\end{tabular}
}
\end{table}
     %%%%%%%%%%%%%%%%%  table  %%%%%%%%%%%%%%%%%

For $n \equiv 2 \pmod 4$,
in order to construct 
new $s$-extremal extremal odd unimodular lattices in dimension $n$,
we consider quasi-twisted self-dual
$[n,n/2]$ codes over $\FF_5$.
Then we found $15$ self-dual
$[38,19]$ codes $C_{38,i}$ over $\FF_5$ 
such that
$A_5(C_{38,i})$ have minimum norm $4$ and kissing number $29260$.
This means that 
$A_5(C_{38,i})$ are 
$s$-extremal extremal odd unimodular lattices in dimension $38$.
For $L=A_5(C_{38,i})$, we calculated by \textsc{Magma}
$\inv(L)_{j}$ $(j=0,1)$, where the results
are listed in Table~\ref{Tab:inv38}.
From~\eqref{eq:G38} and the table, we know that
%$\inv(G_{38})_{j}=\inv(A_5(C_{38,3}))_{j}$ $(j=0,1)$, 
\[
(\inv(G_{38})_{0},\inv(G_{38})_{1})=
(\inv(A_5(C_{38,3})_{0},\inv(A_5(C_{38,3})_{1}),
\]
however, we verified by \textsc{Magma} that these lattices are non-isomorphic.
In addition, 
there are pairs $(i_1,i_2)$ such that 
% $\inv(A_5(C_{38,i_1}))_{j}=\inv(A_5(C_{38,i_2}))_{j}$ $(j=0,1)$.
\[
(\inv(A_5(C_{38,i_1}))_{0},\inv(A_5(C_{38,i_1}))_{1})=
(\inv(A_5(C_{38,i_2}))_{0},\inv(A_5(C_{38,i_2}))_{1}).
\]
For the pairs, we verified by \textsc{Magma} that
$A_5(C_{38,i_1})$ and $A_5(C_{38,i_2})$ are non-isomorphic.
Therefore, we have the following:

\begin{prop}\label{prop:L38}
There are at least $16$ non-isomorphic 
$s$-extremal extremal odd unimodular lattices in dimension $38$.
\end{prop}

For the $15$ codes $C_{38,i}$,
the first rows $r(C_{38,i})$ of the negacirculant matrices $A$ in
the generator matrices of form~\eqref{eq:2}
are listed in Table~\ref{Tab:inv38}.
We verified by \textsc{Magma}
that the codes $C_{38,i}$ have minimum weight $10$ if $i \in
\{1, 3, 4, 5, 8, 9,12,13,15\}$ and the other codes have minimum weight $11$.
Although these codes have minimum weights less than $d_5(38)$,
these codes are useful for constructing $s$-extremal extremal odd unimodular
lattices.

%%%%%%%%%%%%
\subsection{New $s$-extremal extremal odd unimodular lattices in dimension $42$}

Let $L_{42}$ be an extremal odd unimodular lattice in dimension $42$.
By using~\eqref{eq:theta} and~\eqref{eq:theta-S},
one can determine the possible theta series $\theta_{L_{42}}(q)$ and
$\theta_{S(L_{42})}(q)$ as follows.
Since the minimum norm of $L_{42}$ is $4$, 
we have that
\[
a_0=1,
a_1=-84,
a_2=1596 \text{ and }
a_3=-4144,
\]
in~\eqref{eq:theta}.
By considering the coefficients of $q^{\frac{1}{2}}$ and $q^{\frac{5}{2}}$ 
in~\eqref{eq:theta-S},
$a_4$ and $a_5$ are written as:
\[
a_4=2^6\alpha \text{ and }
a_5=-2^{18}\beta,
\]
by using integers $\alpha$ and $\beta$.
Hence, we have that
\begin{equation*}\label{eq:theta42}
\begin{split}
\theta_{L_{42}}(q) &=
1
+( 11844  + 64 \alpha )q^4
+( 1080576  - 768 \alpha  - 262144 \beta )q^5
+ \cdots \text{ and} \\
%\\&+( 41409984 + 2304 \alpha + 9437184 \beta) q^6 + \cdots, \\
\theta_{S(L_{42})}(q)&=
\beta q^{\frac{1}{2}}
+ (\alpha  - 78 \beta )q^{\frac{5}{2}}
+ (265216  - 54 \alpha + 2961 \beta )q^{\frac{9}{2}} + \cdots.
\end{split}
\end{equation*}
An $s$-extremal extremal unimodular lattice
in dimension $42$ is previously known and 
this lattice is denoted by $G_{42}$ in~\cite{Ga04}.
We calculated by \textsc{Magma}
\[
\inv(G_{42})_{0}=22272390
\text{ and }
\inv(G_{42})_{1}= 12633936.
\]

\begin{lem}
Let $L_{42}$ be an extremal odd unimodular lattice in dimension $42$.
Then $L_{42}$ has kissing number $11844$ 
if and only if $L_{42}$ is $s$-extremal.
\end{lem}
\begin{proof}
It is sufficient to show that 
if $L_{42}$ has kissing number $11844$ then $L_{42}$ is $s$-extremal.
Suppose that $L_{42}$ has kissing number $11844$.
By considering the coefficients of $q^{4}$ in $\theta_{L_{42}}(q)$,
we have $\alpha=0$.
By considering the coefficients of $q^{\frac{1}{2}}$ and $q^{\frac{5}{2}}$ in
$\theta_{S(L_{42})}(q)$,
we have $\beta=0$.
The result follows.
\end{proof}

       %%%%%%%%%%%%%%%%%  table  %%%%%%%%%%%%%%%%%
\begin{table}[thp]
\caption{$r(C)$ and $\inv(A_5(C))_{j}$ $(j=0,1)$ for $C=C_{42,i}$}
\label{Tab:inv42}
\centering
\medskip
%{\small
{\footnotesize
%{\scriptsize
\begin{tabular}{c|c|cc}
\noalign{\hrule height1pt}
$i$ & $r(C)$ & $\inv(A_5(C))_{0}$& $\inv(A_5(C))_{1}$\\
\hline
$1$ &$(1,3,2,1,1,2,2,3,2,2,2,4,1,0,4,0,4,3,2,1,1)$& 22228542& 12692400\\
$2$ &$(1,0,3,0,2,0,4,4,0,1,4,2,0,3,4,3,4,2,2,2,0)$& 22242150& 12674256\\
$3$ &$(1,0,2,0,1,3,1,2,1,4,0,0,4,2,3,2,3,0,0,0,0)$& 22247946& 12666528\\
$4$ &$(1,0,4,0,3,1,0,0,1,1,2,2,1,0,1,1,1,1,0,1,1)$& 22258026& 12653088\\
$5$ &$(1,0,1,2,3,0,4,4,0,1,0,2,2,3,2,4,0,2,0,3,1)$& 22261806& 12648048\\
$6$ &$(1,0,4,1,2,2,1,0,4,4,3,4,3,0,1,1,1,1,4,0,1)$& 22261806& 12648048\\
$7$ &$(1,0,0,1,1,1,4,1,1,1,4,1,3,1,1,2,1,4,4,1,0)$& 22262562& 12647040\\
$8$ &$(1,3,1,4,4,0,4,0,0,3,2,2,1,2,0,1,4,2,0,4,1)$& 22263318& 12646032\\
$9$ &$(1,2,2,0,0,2,2,1,4,0,1,1,1,3,1,3,2,3,0,4,2)$& 22264830& 12644016\\
$10$ &$(1,0,1,3,0,0,1,4,2,3,2,0,0,0,3,0,3,0,1,4,3)$& 22267602& 12640320\\
$11$ &$(1,4,3,3,1,3,4,4,4,4,4,2,0,0,2,1,2,0,1,3,4)$& 22270626& 12636288\\
$12$ &$(1,3,2,0,3,0,0,3,2,2,1,1,2,4,2,3,4,2,4,3,2)$& 22271634& 12634944\\
$13$ &$(1,3,4,2,2,4,4,0,1,3,2,0,3,0,0,1,0,3,0,1,2)$& 22272138& 12634272\\
$14$ &$(1,0,1,4,1,0,0,2,0,3,2,2,1,2,0,4,1,0,1,0,1)$& 22273902& 12631920\\
$15$ &$(1,0,2,4,2,2,3,2,1,3,0,0,1,0,3,0,1,2,2,3,3)$& 22273902& 12631920\\
$16$ &$(1,0,4,0,2,1,4,4,0,3,2,0,4,3,4,1,3,1,2,0,4)$& 22274658& 12630912\\
$17$ &$(1,0,0,1,1,1,3,1,1,1,4,1,1,2,1,2,0,0,4,0,0)$& 22275918& 12629232\\
$18$ &$(1,1,1,0,0,0,2,0,2,2,0,2,1,0,4,2,4,1,1,0,4)$& 22278438& 12625872\\
$19$ &$(1,3,3,4,2,4,0,4,2,4,1,1,3,3,3,3,0,3,0,0,1)$& 22280706& 12622848\\
$20$ &$(1,4,3,3,0,0,2,2,2,4,0,0,3,4,4,0,4,4,3,0,2)$& 22284486& 12617808\\
$21$ &$(1,2,2,1,2,2,1,1,3,2,3,0,0,1,4,2,0,1,0,2,1)$& 22285242& 12616800\\
$22$ &$(1,0,3,4,0,2,0,4,0,1,1,4,0,0,0,4,4,3,3,3,1)$& 22286754& 12614784\\
$23$ &$(1,1,1,2,2,3,1,2,1,3,2,1,3,3,1,2,1,4,3,1,3)$& 22300866& 12595968\\
$24$ &$(1,0,2,1,0,3,4,3,0,3,0,3,2,3,4,2,2,3,0,3,1)$& 22304142& 12591600\\
$25$ &$(1,0,3,3,3,4,3,0,2,1,4,0,2,1,0,3,0,2,3,2,3)$& 22311450& 12581856\\
$26$ &$(1,1,2,1,4,1,1,2,3,1,4,2,1,0,2,2,2,4,4,1,3)$& 22320774& 12569424\\
$27$ &$(1,2,3,3,3,2,1,3,2,1,3,0,4,2,0,2,3,2,0,4,4)$& 22338162& 12546240\\
$28$ &$(1,0,2,1,0,4,1,2,1,0,0,4,3,3,1,3,1,1,2,1,0)$& 22339674& 12544224\\
$29$ &$(1,0,1,0,1,3,4,0,3,1,2,0,1,1,0,2,4,3,4,2,4)$& 22365378& 12509952\\
$30$ &$(1,0,2,1,4,2,4,0,3,1,2,0,1,1,0,1,0,1,0,0,3)$& 22391586& 12475008\\
\noalign{\hrule height1pt}
\end{tabular}
}
\end{table}
     %%%%%%%%%%%%%%%%%  table  %%%%%%%%%%%%%%%%%

By considering quasi-twisted codes,
we found $30$ self-dual $[42,21]$ codes $C_{42,i}$ over $\FF_5$ 
such that
$A_5(C_{42,i})$ have minimum norm $4$ and kissing number $11844$.
This means that 
$A_5(C_{42,i})$ are 
$s$-extremal extremal odd unimodular lattices in dimension $42$.
For $L=A_5(C_{42,i})$, we calculated by \textsc{Magma}
$\inv(L)_{j}$ $(j=0,1)$, where the results
are listed in Table~\ref{Tab:inv42}.
There are pairs $(i_1,i_2)$ such that 
%$\inv(A_5(C_{42,i_1}))_{j}=\inv(A_5(C_{42,i_2}))_{j}$ $(j=0,1)$.
\[
(\inv(A_5(C_{42,i_1}))_{0},\inv(A_5(C_{42,i_1}))_{1}) =
(\inv(A_5(C_{42,i_2}))_{0},\inv(A_5(C_{42,i_2}))_{1}).
\]
For the pairs, we verified by \textsc{Magma} that
$A_5(C_{42,i_1})$ and $A_5(C_{42,i_2})$ are non-isomorphic.
Therefore, we have the following:

\begin{prop}\label{prop:L42}
There are at least $31$ non-isomorphic 
$s$-extremal extremal odd unimodular lattices in dimension $42$.
\end{prop}

For the $30$ codes $C_{42,i}$,
the first rows $r(C_{42,i})$ of the negacirculant matrices $A$ in
the generator matrices of form~\eqref{eq:2}
are listed in Table~\ref{Tab:inv42}.
We verified by \textsc{Magma}
that the codes $C_{42,i}$ have minimum weight $10$ if $i \in
\{3,5,6,22,26,28,29\}$,
the code $C_{42,19}$ have minimum weight $11$
and the other codes have minimum weight $12$.
Although these codes have minimum weights less than $d_5(42)$,
these codes are useful for constructing $s$-extremal extremal odd unimodular
lattices.

%%%%%%%%%%%%%%%%%%
\subsection{Summary of  the existence of  $s$-extremal unimodular lattices
with minimum norm $4$}

In Section~\ref{sec:44} and this section, 
we constructed $15$, $50$, $30$ and $153$  new non-isomorphic 
$s$-extremal extremal unimodular lattices in dimensions
$38$, $40$, $42$ and $44$, respectively.
As a summary,  we list  in Table~\ref{Tab:2} the current information on 
the number $N(n)$ of non-isomorphic
$s$-extremal unimodular lattices $L$ in dimension $n$
with $\min(L)=4$.
The table updates the table given in~\cite[p.~148]{Ga07}
and Table~2 in~\cite{H11}.
%We remark that the existence of
%a unimodular lattice $L$ in dimension $37$ having minimum norm $4$
%is still unknown  (see~\cite{lattice-datebase}).

       %%%%%%%%%%%%%%%%%  table  %%%%%%%%%%%%%%%%%
\begin{table}[thbp]
\caption{Existence of $s$-extremal unimodular lattices $L$ with $\min(L)=4$}
\label{Tab:2}
\centering\medskip
{\small
%{\footnotesize
%{\scriptsize
\begin{tabular}{c|c|c||c|c|c}
\noalign{\hrule height1pt}
$n$ & $N(n)$ & \multicolumn{1}{c||}{References} &
$n$ & $N(n)$ & \multicolumn{1}{c}{References} \\
\hline
32 & 5       &\cite{CS98} &42 & $\ge 31$ & Proposition~\ref{prop:L42}\\
36 & $\ge 19$ &\cite{H14} &43 & $\ge 1$ &\cite{H11}\\
37 & ?       &             &44 & $\ge 156$ & Proposition~\ref{prop:L}\\
38 & $\ge 16$ & Proposition~\ref{prop:L38} &45 & ? &    \\
39 & $\ge 1$ &\cite{GH99} &46 & $\ge 1$ &\cite{H03} \\
40 & $\ge 51$& Proposition~\ref{prop:L40}\
                           &47 & $\ge 1$ &\cite{H03} \\
41 & $\ge 1$ &\cite{H11} &&&\\
\noalign{\hrule height1pt}
\end{tabular}
}
\end{table}
     %%%%%%%%%%%%%%%%%  table  %%%%%%%%%%%%%%%%%

\bigskip
\noindent
\textbf{Acknowledgments.}
This work was supported by JSPS KAKENHI Grant Number 19H01802.

%%%%%%%%%%%%%%%%%%%  References  %%%%%%%%%%%%%%%%%%%%%%%%

%%%%%%%%%%%%%%%%%%%%%%%%%%%%%%%%%

\begin{landscape}

       %%%%%%%%%%%%%%%%%  table  %%%%%%%%%%%%%%%%%
\begin{table}[thbp]
\caption{Self-dual $[44,22]$ codes $C_{44,i}$ over $\FF_5$}
\label{Tab:44}
\centering\medskip
{\small
%{\footnotesize
%{\scriptsize
\begin{tabular}{c|c|c||c|c|c}
\noalign{\hrule height1pt}
$i$ & $r_A(C_{44,i})$ & $r_B(C_{44,i})$ &
$i$ & $r_A(C_{44,i})$ & $r_B(C_{44,i})$ \\
\hline
 1&$(1,0,0,4,4,4,0,2,3,1,1)$&$(2,0,1,0,2,0,0,4,3,4,0)$&
26&$(1,0,0,3,2,4,3,2,0,3,1)$&$(2,1,0,1,4,4,1,3,2,3,0)$\\
 2&$(1,0,0,1,2,2,2,3,2,0,4)$&$(4,0,3,4,1,2,0,1,1,2,2)$&
27&$(1,0,0,1,2,1,3,0,4,2,1)$&$(1,1,3,2,3,4,3,2,3,3,4)$\\
 3&$(1,0,0,1,1,3,4,1,1,4,0)$&$(1,3,3,2,2,0,2,2,0,3,2)$&
28&$(1,0,0,4,0,0,1,1,1,0,2)$&$(3,4,3,3,2,3,4,1,2,2,2)$\\
 4&$(1,0,0,4,0,2,1,4,0,1,2)$&$(3,1,0,4,4,1,3,2,0,4,2)$&
29&$(1,0,0,0,4,4,1,1,0,2,4)$&$(1,0,0,3,0,3,0,2,4,4,2)$\\
 5&$(1,0,0,0,3,1,1,4,1,0,1)$&$(1,4,3,2,4,2,0,4,2,0,2)$&
30&$(1,0,0,1,2,1,4,2,1,4,3)$&$(3,3,0,1,1,4,1,4,4,1,4)$\\
 6&$(1,0,0,1,1,0,4,1,2,2,3)$&$(2,1,0,1,0,3,1,0,2,4,4)$&
31&$(1,0,0,0,2,3,3,4,2,0,4)$&$(4,0,3,4,1,2,0,1,4,1,4)$\\
 7&$(1,0,0,3,3,0,3,1,1,4,0)$&$(1,3,3,2,2,0,2,0,1,0,4)$&
32&$(1,0,0,3,1,3,3,1,1,4,0)$&$(1,3,3,2,2,0,2,0,4,0,0)$\\
 8&$(1,0,0,4,1,1,4,1,0,1,3)$&$(1,4,1,0,0,2,1,4,3,4,3)$&
33&$(1,0,0,2,4,1,2,3,4,4,2)$&$(3,2,2,3,4,2,2,0,3,0,3)$\\
 9&$(1,0,0,0,1,4,4,3,4,0,4)$&$(4,0,1,0,2,1,3,1,2,2,2)$&
34&$(1,0,0,0,1,3,2,0,4,4,2)$&$(1,4,1,1,3,1,0,0,1,2,2)$\\
10&$(1,0,0,0,2,3,3,3,2,4,2)$&$(2,1,0,1,3,0,4,2,1,4,1)$&
35&$(1,0,0,4,1,1,2,2,2,4,2)$&$(1,3,1,3,2,0,4,2,2,4,2)$\\
11&$(1,0,0,3,3,2,4,1,1,4,1)$&$(2,2,1,0,0,1,0,1,1,3,0)$&
36&$(1,0,0,2,3,4,4,1,4,2,2)$&$(3,3,1,4,3,3,3,4,0,3,4)$\\
12&$(1,0,0,3,0,3,0,1,0,0,3)$&$(0,4,2,1,0,1,1,3,3,0,3)$&
37&$(1,0,0,1,4,0,4,4,2,0,4)$&$(4,0,3,4,1,2,0,2,3,4,3)$\\
13&$(1,0,0,4,3,4,0,4,2,2,3)$&$(2,3,0,4,2,3,1,4,4,2,0)$&
38&$(1,0,0,2,3,3,3,0,3,0,2)$&$(4,1,2,2,4,4,0,0,1,4,0)$\\
14&$(1,0,0,1,2,1,3,1,1,4,0)$&$(1,3,3,2,2,0,2,0,3,1,2)$&
39&$(1,0,0,4,0,1,3,1,3,0,1)$&$(1,4,1,1,1,0,0,0,2,1,1)$\\
15&$(1,0,0,2,2,1,1,4,0,2,2)$&$(3,0,4,0,0,0,4,4,3,3,3)$&
40&$(1,0,0,1,1,0,3,0,1,4,1)$&$(2,2,1,0,0,1,0,0,3,2,4)$\\
16&$(1,0,0,0,4,3,1,1,1,4,1)$&$(2,2,1,0,0,1,0,3,3,4,2)$&
41&$(1,0,0,0,4,0,3,3,3,2,4)$&$(3,2,0,1,3,3,1,3,0,2,2)$\\
17&$(1,0,0,2,2,0,0,2,2,0,2)$&$(0,0,1,2,2,3,4,3,2,1,0)$&
42&$(1,0,0,4,3,1,2,4,0,2,2)$&$(3,0,4,0,0,0,4,3,2,0,0)$\\
18&$(1,0,0,2,3,0,0,3,2,0,4)$&$(4,0,3,4,1,2,0,2,4,4,2)$&
43&$(1,0,0,4,2,2,1,1,1,0,2)$&$(3,4,3,3,2,3,4,0,4,0,2)$\\
19&$(1,0,0,4,2,4,0,1,0,4,0)$&$(3,2,4,3,3,4,4,4,4,0,2)$&
44&$(1,0,0,4,1,0,4,4,2,2,4)$&$(4,3,0,0,3,4,2,2,1,1,0)$\\
20&$(1,0,0,3,3,3,2,4,1,2,3)$&$(0,0,4,2,2,0,0,2,2,3,4)$&
45&$(1,0,0,4,4,0,4,3,4,3,2)$&$(2,2,1,2,1,3,1,2,3,3,1)$\\
21&$(1,0,0,4,4,1,2,1,0,4,3)$&$(0,3,2,2,4,3,1,2,3,2,0)$&
46&$(1,0,0,0,4,4,3,0,1,1,4)$&$(3,3,3,1,4,1,2,4,3,4,3)$\\
22&$(1,0,0,0,0,3,0,3,4,0,4)$&$(4,0,1,0,2,1,3,2,1,4,1)$&
47&$(1,0,0,0,4,0,2,1,4,2,2)$&$(3,3,1,4,3,3,3,0,4,4,2)$\\
23&$(1,0,0,2,3,1,4,4,1,2,3)$&$(0,0,4,2,2,0,0,0,2,4,3)$&
48&$(1,0,0,1,4,0,4,2,4,2,2)$&$(3,3,1,4,3,3,3,1,0,0,3)$\\
24&$(1,0,0,0,2,3,0,2,1,0,3)$&$(1,0,4,3,0,2,3,2,4,4,4)$&
49&$(1,0,0,4,2,0,4,1,1,3,0)$&$(4,4,3,4,3,2,2,0,0,1,1)$\\
25&$(1,0,0,4,2,0,4,0,3,0,2)$&$(4,1,2,2,4,4,0,1,1,2,4)$&
50&$(1,0,0,3,3,2,1,0,4,0,4)$&$(1,2,2,4,1,4,1,3,3,4,4)$\\
% old data 50&$(1,0,0,3,1,3,3,4,3,3,0)$&$(0,4,1,4,4,4,2,3,1,4,4)$\\
\noalign{\hrule height1pt}
\end{tabular}
}
\end{table}
     %%%%%%%%%%%%%%%%%  table  %%%%%%%%%%%%%%%%%

  %%%%%%%%%%%%%%%%%  table  %%%%%%%%%%%%%%%%%
\begin{table}[thbp]
\caption{Self-dual $[40,20,12]$ codes $C_{40,i}$ over $\FF_5$}
\label{Tab:40}
\centering\medskip
{\small
%{\footnotesize
%{\scriptsize
\begin{tabular}{c|c|c||c|c|c}
\noalign{\hrule height1pt}
$i$ & $r_A(C_{40,i})$ & $r_B(C_{40,i})$ &
$i$ & $r_A(C_{40,i})$ & $r_B(C_{40,i})$ \\
\hline
 1 &$(1,0,0,3,0,3,2,3,3,1)$&$(4,3,0,3,2,2,3,3,1,1)$ &
26 &$(1,0,0,3,3,2,1,3,1,3)$&$(2,3,2,2,2,4,2,1,1,2)$ \\
 2 &$(1,0,0,3,2,3,3,3,4,1)$&$(4,3,0,3,2,2,1,2,2,0)$ &
27 &$(1,0,0,4,0,0,0,1,1,3)$&$(2,3,2,2,2,4,2,0,0,1)$ \\
 3 &$(1,0,0,2,4,0,0,3,2,3)$&$(2,3,2,2,2,4,4,4,2,3)$ &
28 &$(1,0,0,3,3,1,3,3,1,3)$&$(2,3,2,2,2,4,2,0,0,4)$ \\
 4 &$(1,0,0,0,3,1,0,2,1,3)$&$(2,3,2,2,2,4,2,3,3,4)$ &
29 &$(1,0,0,4,3,4,3,0,1,3)$&$(2,3,2,2,2,4,2,1,4,1)$ \\
 5 &$(1,0,0,2,4,4,0,2,3,1)$&$(4,3,0,3,2,2,2,4,1,0)$ &
30 &$(1,0,0,2,2,1,2,0,2,3)$&$(2,3,2,2,2,4,4,0,2,4)$ \\
 6 &$(1,0,0,2,4,3,3,2,1,3)$&$(2,3,2,2,2,4,4,4,2,3)$ &
31 &$(1,0,0,4,1,4,4,4,3,1)$&$(4,3,0,3,2,2,0,1,2,4)$ \\
 7 &$(1,0,0,2,0,2,4,3,3,1)$&$(4,3,0,3,2,2,3,3,4,2)$ &
32 &$(1,0,0,2,2,4,1,1,3,1)$&$(4,3,0,3,2,2,2,1,2,1)$ \\
 8 &$(1,0,0,0,4,1,4,3,4,1)$&$(4,3,0,3,2,2,1,1,2,1)$ &
33 &$(1,0,0,0,4,3,4,3,2,3)$&$(2,3,2,2,2,4,1,2,0,2)$ \\
 9 &$(1,0,0,3,2,1,3,1,2,3)$&$(2,3,2,2,2,4,2,3,4,1)$ &
34 &$(1,0,0,3,2,4,2,1,2,3)$&$(2,3,2,2,2,4,0,1,3,0)$ \\
10 &$(1,0,0,4,3,0,3,4,1,3)$&$(2,3,2,2,2,4,0,0,1,1)$ &
35 &$(1,0,0,4,3,3,3,1,4,1)$&$(4,3,0,3,2,2,1,3,2,4)$ \\
11 &$(1,0,0,0,4,0,3,3,2,3)$&$(2,3,2,2,2,4,0,1,0,3)$ &
36 &$(1,0,0,0,1,0,1,3,1,3)$&$(2,3,2,2,2,4,4,2,1,0)$ \\
12 &$(1,0,0,0,0,1,2,2,4,1)$&$(1,3,1,0,1,3,1,0,4,2)$ &
37 &$(1,0,0,3,4,0,3,2,1,3)$&$(2,3,2,2,2,4,1,2,0,2)$ \\
13 &$(1,0,0,0,4,3,4,1,3,1)$&$(4,3,0,3,2,2,3,2,4,0)$ &
38 &$(1,0,0,0,1,1,4,2,2,3)$&$(2,3,2,2,2,4,3,3,2,0)$ \\
14 &$(1,0,0,4,0,0,2,2,4,1)$&$(4,3,0,3,2,2,4,3,2,4)$ &
39 &$(1,0,0,3,0,1,1,1,1,3)$&$(2,3,2,2,2,4,4,2,0,0)$ \\
15 &$(1,0,0,2,0,2,4,0,3,1)$&$(4,3,0,3,2,2,3,3,0,3)$ &
40 &$(1,0,0,1,3,4,1,0,2,3)$&$(2,3,2,2,2,4,3,3,3,0)$ \\
16 &$(1,0,0,4,0,0,0,2,3,1)$&$(4,3,0,3,2,2,3,1,0,1)$ &
41 &$(1,0,0,1,3,0,0,2,2,3)$&$(2,3,2,2,2,4,0,2,0,4)$ \\
17 &$(1,0,0,4,1,3,3,1,1,3)$&$(2,3,2,2,2,4,0,3,1,1)$ &
42 &$(1,0,0,3,1,3,1,0,2,3)$&$(2,3,2,2,2,4,0,4,3,3)$ \\
18 &$(1,0,0,0,1,2,3,3,3,1)$&$(4,3,0,3,2,2,3,3,0,0)$ &
43 &$(1,0,0,2,0,0,4,1,3,1)$&$(4,3,0,3,2,2,1,3,1,3)$ \\
19 &$(1,0,0,2,3,4,2,4,3,1)$&$(4,3,0,3,2,2,1,1,3,1)$ &
44 &$(1,0,0,4,1,4,1,0,1,3)$&$(2,3,2,2,2,4,4,3,3,3)$ \\
20 &$(1,0,0,0,3,3,2,3,1,3)$&$(2,3,2,2,2,4,3,1,0,1)$ &
45 &$(1,0,0,2,4,3,1,0,2,3)$&$(2,3,2,2,2,4,0,3,2,4)$ \\
21 &$(1,0,0,3,4,1,4,4,1,3)$&$(2,3,2,2,2,4,4,4,1,4)$ &
46 &$(1,0,0,2,0,0,4,1,3,1)$&$(4,3,0,3,2,2,0,0,1,2)$ \\
22 &$(1,0,0,1,0,1,2,1,1,3)$&$(2,3,2,2,2,4,1,4,3,2)$ &
47 &$(1,0,0,1,1,2,1,0,1,3)$&$(2,3,2,2,2,4,3,0,4,0)$ \\
23 &$(1,0,0,3,1,4,3,1,1,3)$&$(2,3,2,2,2,4,0,3,4,1)$ &
48 &$(1,0,0,4,1,3,0,1,1,3)$&$(2,3,2,2,2,4,3,1,3,1)$ \\
24 &$(1,0,0,3,2,3,0,3,1,3)$&$(2,3,2,2,2,4,4,0,4,3)$ &
49 &$(1,0,0,2,3,2,4,3,1,3)$&$(2,3,2,2,2,4,0,2,0,4)$ \\
25 &$(1,0,0,2,1,1,3,3,3,1)$&$(4,3,0,3,2,2,3,0,2,2)$ &
50 &$(1,0,0,4,4,3,4,2,2,3)$&$(2,3,2,2,2,4,1,4,1,0)$ \\
\noalign{\hrule height1pt}
\end{tabular}
}
\end{table}
     %%%%%%%%%%%%%%%%%  table  %%%%%%%%%%%%%%%%%

\end{landscape}

\end{document}